\documentclass[runningheads]{llncs}

\usepackage[table,dvipsnames]{xcolor}
\usepackage{hyperref}
\usepackage{caption}
\usepackage{amsmath}
\usepackage{amssymb}
\usepackage{amsfonts}
\usepackage{marvosym}
\usepackage{calc}
\usepackage[capitalize]{cleveref}
\usepackage{xspace}
\usepackage{tikz}
\usepackage{subcaption}
\usepackage{wrapfig}
\usepackage{ifthen}
\usepackage{stmaryrd}

\usepackage[colorinlistoftodos,prependcaption,textsize=small]{todonotes}
\usepackage{booktabs}
\usepackage{paralist}

\usepackage{orcidlink}
\usepackage[ruled,noend,linesnumbered]{algorithm2e}

\newif\ifTR
\TRtrue

\usepackage{afterpage}

\usetikzlibrary{shapes.misc, positioning}
\usetikzlibrary{decorations.pathreplacing}
\usetikzlibrary{arrows,automata}
\usetikzlibrary{patterns.meta}
\usetikzlibrary{calc}

\Crefname{algocf}{Algorithm}{Algorithms}

\usetikzlibrary{shapes.geometric}
\usetikzlibrary{backgrounds}
\usetikzlibrary{fit}

\usetikzlibrary{automata}
\usetikzlibrary{arrows.meta}
\usetikzlibrary{bending}
\usetikzlibrary{shapes.callouts}
\usetikzlibrary{quotes}
\usetikzlibrary{positioning}
\usetikzlibrary{calc}
\usetikzlibrary{matrix}

\definecolor{spotblue}{RGB}{31,120,180}
\definecolor{spotpink}{RGB}{255,77,160}
\definecolor{spotorange}{RGB}{255,127,0}
\definecolor{spotpurple}{RGB}{106,61,154}
\definecolor{spotgreen}{RGB}{51,160,44}
\definecolor{spotred}{RGB}{227,26,28}
\definecolor{spotyellowish}{RGB}{196,196,0}
\definecolor{spotgray}{RGB}{80,80,80}
\definecolor{spotlight blue}{RGB}{107,246,255}
\definecolor{spotlight pink}{RGB}{255,154,255}
\definecolor{spotlight orange}{RGB}{255,156,103}
\definecolor{spotlight purple}{RGB}{178,164,255}
\definecolor{spotlight green}{RGB}{167,237,121}
\definecolor{spotlight red}{RGB}{255,104,104}
\definecolor{spotlight yellowish}{RGB}{255,224,64}
\definecolor{spotlight gray}{RGB}{192,192,144}

\tikzset{
  >={Stealth[round,bend]},
}

\tikzstyle{automaton}=[
  semithick,shorten >=0pt,>={Stealth[round,bend]},
  node distance=1.5cm,
  initial text=,
  every initial by arrow/.style={every node/.style={inner sep=0pt}},
  every state/.style={
    align=center,
    fill=white,
    minimum size=7.5mm,
    inner sep=0pt,
    execute at begin node=\strut,
  },%
]
\tikzset{
  scc/.style={draw=gray,fill=black!10,rounded corners=2mm},
  lstate/.style={state},
  }

\tikzstyle{smallautomaton}=[
  automaton,
  node distance=7mm,
  every state/.style={minimum size=4mm,fill=white,inner sep=1.5pt}
]

\tikzstyle{mediumautomaton}=[
  automaton,
  node distance=1cm,
  every state/.style={minimum size=6mm,fill=white,inner sep=2pt}
]

\tikzstyle{cstate}=[state,capsule,text width=,inner xsep=-5pt]
\tikzstyle{dot}=[fill=black,circle,minimum size=4pt,inner sep=0]

\makeatletter
\tikzoption{initial angle}{\tikzaddafternodepathoption{\def\tikz@initial@angle{#1}}}
\makeatother

\tikzstyle{initial overlay}=[every initial by arrow/.append style={overlay}]

\tikzstyle{unreachable} = [densely dotted]

\tikzstyle{acclabel} = [
  fill=yellow!20,
  rounded corners=3pt,
  draw=black,
  inner ysep=0pt,
  inner xsep=3pt,
]
\tikzstyle{ltllabel} = [acclabel,fill=darkgreen!20]

\tikzstyle{namelabel} = [
  execute at begin node = {$($},%
  execute at end node = {$)$}
]

\tikzstyle{matrix of states} = [
matrix of nodes,
every node/.style={state},
column sep=.8cm,
row sep=.8cm,
execute at empty cell={
  \node[transparent,name=
    \tikzmatrixname-\the\pgfmatrixcurrentrow-\the\pgfmatrixcurrentcolumn]{};]},
]

\tikzstyle{accset}=[
  rectangle,inner sep=1.5pt,     
	rounded corners=1.5mm,
  draw=none, 
  solid,
  collacc0, text=white,
  thin,
  anchor=center,
  minimum size=3.3mm,
  text width={}, 
  font=\bfseries\sffamily\footnotesize
]
\tikzstyle{accsquare}=[accset,rectangle,inner sep=1.9pt,rounded corners=0pt]

\tikzset{
  collacc0/.style={fill=spotblue},
  collacc1/.style={fill=spotpink},
  collacc2/.style={fill=spotorange},
  collacc3/.style={fill=spotpurple},
  collacc4/.style={fill=spotgreen},
  collacc5/.style={fill=spotred},
  collacc6/.style={fill=spotyellowish,draw=black,text=black},
  collacc7/.style={fill=spotgray},
  collacc8/.style={fill=spotlight blue,draw=black,text=black},
  collacc9/.style={fill=spotlight pink},
  collacc10/.style={fill=spotlight orange},
  collacc11/.style={fill=spotlight purple},
  collacc12/.style={fill=spotlight green},
  collacc13/.style={fill=spotlight red},
  collacc14/.style={fill=spotlight yellowish},
  collacc15/.style={fill=spotlight gray},
}

\pgfkeyssetvalue{/sacc/where}{center}
\tikzset{
  sacc where/.code={
    \pgfkeyssetvalue{/sacc/where}{#1}
  }
}

\tikzset{
  l/.pic={\node[outer sep=2pt] {#1};},%
  acc/.pic={\node[accset,collacc#1]{\upshape #1};},%
  accsq/.pic={\node[accsquare,collacc#1]{#1};},%
  eacc/.pic={\node[accset,collacc#1]{\emptyacc};},%
  eaccsq/.pic={\node[accsquare,collacc#1]{\emptyacc};},%
  sacc/.style = {
    append after command=
      {pic at (\tikzlastnode.\pgfkeysvalueof{/sacc/where}) {acc=#1}}
  },
  saccsq/.style = {
    append after command=
      {pic at (\tikzlastnode.\pgfkeysvalueof{/sacc/where}) {accsq=#1}}
  },
  esacc/.style = {
    append after command=
      {pic at (\tikzlastnode.\pgfkeysvalueof{/sacc/where}) {eacc=#1}}
  },
  esaccsq/.style = {
    append after command=
      {pic at (\tikzlastnode.\pgfkeysvalueof{/sacc/where}) {eaccsq=#1}}
  },
  pics/cacc/.style 2 args={%
    code={\node[accset,collacc#1]{#2};}%
  },%
  pics/caccsq/.style 2 args={%
      code={\node[accsquare,collacc#1]{#2};}%
    },%
  csacc/.style 2 args = {%
      append after command=%
        {pic at (\tikzlastnode.\pgfkeysvalueof{/sacc/where}) {cacc={#1}{#2}}}%
  },%
  csaccsq/.style 2 args = {%
      append after command=%
        {pic at (\tikzlastnode.\pgfkeysvalueof{/sacc/where}) {caccsq={#1}{#2}}}%
  },%
}

\def\markbaseline{-.33em}
\def\tacc#1{\tikz[baseline=\markbaseline]\pic{acc=#1};\xspace}
\def\tcacc#1#2{\tikz[baseline=\markbaseline]\pic{cacc={#1}{\upshape #2}};\xspace}

\def\emptyacc{\phantom{0}}

\tikzstyle{removed} = [opacity=0, overlay]
\tikzstyle{SCC} = [
  rounded corners,
  draw=black!50, very thin,
  fill=black!7
]
\tikzstyle{trivial} = [dashed]
\tikzstyle{sccname} = [red,anchor=south east,outer xsep=13pt]

\newcommand{\ol}[1]{\textcolor{blue}{\ifmmode \text{[OL: #1]}\else [OL: #1] \fi}}
\newcommand{\vh}[1]{\todo[linecolor=red,backgroundcolor=red!25,bordercolor=red]{VH: #1}}
\newcommand{\ly}[1]{\textcolor{orange}{\ifmmode \text{[LY: #1]}\else [LY: #1] \fi}}
\newcommand{\lh}[1]{\textcolor{pink}{\ifmmode \text{[LH: #1]}\else [LH: #1] \fi}}
\newcommand{\nm}[1]{\textcolor{brown}{\ifmmode \text{[NM: #1]}\else [NM: #1] \fi}}

\newcommand{\buchi}[0]{B\"{u}chi\xspace}

\newcommand{\aut}[0]{\mathcal{A}}
\newcommand{\but}[0]{\mathcal{B}}

\newcommand{\alphabet}[0]{\Sigma}
\newcommand{\word}[0]{w}
\newcommand{\wordof}[1]{\word_{#1}}
\newcommand{\trans}[0]{\delta}
\newcommand{\acc}[0]{\mathit{Acc}}
\newcommand{\acctrans}[0]{\trans_{\!\acc}}
\newcommand{\ltr}[1]{\mathrel{\raisebox{0pt}[0pt][0pt]{\ensuremath{\xrightarrow{\raisebox{-5mm}{$\scriptstyle#1$}}}}}}

\newcommand{\lang}[0]{\mathcal{L}}
\newcommand{\langof}[1]{\lang(#1)}

\newcommand{\restrof}[2]{#1 \raisebox{-.5ex}{$|$}_{#2}}
\newcommand{\colours}[0]{\Gamma}
\newcommand{\colouring}[0]{\mathsf{p}}
\newcommand{\colouringof}[1]{\colouring(#1)}

\newcommand{\infsymb}[0]{\mathit{inf}}

\newcommand{\inftof}[1]{\infsymb(#1)}
\newcommand{\Inf}[0]{\mathit{Inf}}
\newcommand{\Infof}[1]{\Inf(#1)}

\newcommand{\Fin}[0]{\mathit{Fin}}
\newcommand{\Finof}[1]{\Fin(#1)}

\newcommand{\alg}[0]{\mathtt{Alg}}

\newcommand{\GetSucc}[0]{\mathtt{Succ}}

\newcommand{\iw}{\mathcal{W}}

\newcommand{\nac}{\mathcal{N}}

\newcommand{\pspace}[0]{\textsc{PSpace}\xspace}

\newcommand{\kofola}[0]{\textsc{Kofola}\xspace}
\newcommand{\kofolaslice}[0]{\textsc{KofolaSlice}\xspace}
\newcommand{\kofolaold}[0]{\textsc{KofolaOld}\xspace}

\newcommand{\spot}[0]{\textsc{Spot}\xspace}
\newcommand{\spotdet}[0]{\textsc{Spot(Det)}\xspace}
\newcommand{\spotforq}[0]{\textsc{Spot(Forq)}\xspace}

\newcommand{\rabit}[0]{\textsc{Rabit}\xspace}
\newcommand{\bait}[0]{\textsc{Bait}\xspace}
\newcommand{\forklift}[0]{\textsc{Forklift}\xspace}
\newcommand{\owl}[0]{\textsc{Owl}\xspace}
\newcommand{\autohyper}[0]{\textsc{AutoHyper}\xspace}
\newcommand{\automizer}[0]{\textsc{Automizer}\xspace}
\newcommand{\pecan}[0]{\textsc{Pecan}\xspace}
\newcommand{\concur}[0]{\textsc{Concur}\xspace}
\newcommand{\ranker}[0]{\textsc{Ranker}\xspace}
\newcommand{\ldbaforltl}[0]{\textsc{Ldba4Ltl}\xspace}
\newcommand{\seminator}[0]{\textsc{Seminator}\xspace}
\newcommand{\stateofbuchi}[0]{\textsc{Sobc}\xspace}

\newcommand{\RGBcircle}[1]{\ensuremath{{\color[RGB]{#1}\bullet}}}

\newcommand{\lnref}[1]{Line~\ref{#1}}

\newcommand{\autohyperbenchc}[0]{$\textsc{AutoHyper}_{\mathcal{C}}$\xspace}
\newcommand{\automizerbenchc}[0]{$\textsc{Automizer}_{\mathcal{C}}$\xspace}
\newcommand{\pecanbenchc}[0]{$\textsc{Pecan}_{\mathcal{C}}$\xspace}
\newcommand{\sones}[0]{\textsc{S1S}\xspace}
\newcommand{\autohyperbenchi}[0]{$\textsc{AutoHyper}_{\mathcal{I}}$\xspace}
\newcommand{\automizerbenchi}[0]{$\textsc{Automizer}_{\mathcal{I}}$\xspace}
\newcommand{\pecanbenchi}[0]{$\textsc{Pecan}_{\mathcal{I}}$\xspace}

\newcommand{\redfattimes}[0]{\textcolor{red}{$\pmb{\pmb{\times}}$}}

\usepackage{todonotes}

\title{\kofola 1.0: A~Modular Approach to $\omega$-Regular Complementation and
Inclusion Checking\ifTR\\(Technical Report)\fi}

\titlerunning{\kofola 1.0: Modular $\omega$-Regular Complementation and
Inclusion Checking}

\author{
  Ondrej Alexaj\orcidlink{0009-0009-3994-4563}\Letter\inst{1} \and
  Vojtěch Havlena\orcidlink{0000-0003-4375-7954}\Letter\inst{1} \and
  Lukáš Holík\orcidlink{0000-0001-6957-1651}\Letter\inst{1,2} \and\\
  Ondřej Lengál\orcidlink{0000-0002-3038-5875}\Letter\inst{1} \and
  Yong Li\orcidlink{0000-0002-7301-9234}\Letter\inst{3} \and
  Nicolas Mazzocchi\orcidlink{0000-0001-6425-5369}\Letter\inst{4}
} 

\authorrunning{O. Alexaj, V. Havlena, L. Holík, O. Lengál, Y. Li, N. Mazzocchi}

\institute{
  Brno University of Technology, Brno, Czech Republic \and
  Aalborg University, Aalborg, Denmark \and
  Key Lab. of System Software (Chinese Academy of Sciences), ISCAS, Beijing, PRC \and
  Slovak University of Technology in Bratislava, Bratislava, Slovakia
  \email{xalexa09@stud.fit.vut.cz, ihavlena@fit.vut.cz, holik@fit.vut.cz, lengal@fit.vut.cz, liyong@ios.ac.cn, nicolas.mazzocchi@stuba.sk}
}

\begin{document}
\maketitle

\begin{abstract}
We present \kofola, an efficient tool for complementation and inclusion
checking of \buchi automata, two central tasks in automata-theoretic
verification with applications in model checking, monitoring, and theorem
proving. \kofola implements a state-of-the-art modular complementation
framework that decomposes the input automaton into strongly connected
components and applies to each component a complementation algorithm tailored
to its structural properties.
Building on this modular construction, \kofola also provides modular
\emph{inclusion checking} with new heuristics. A key ingredient is a new
on-the-fly emptiness-checking algorithm for the simple generalized Rabin pair
condition produced by our complementation, allowing the search to terminate as
soon as the explored state space suffices. Empirical evaluation shows that \kofola is
highly competitive with state-of-the-art complementation and inclusion-checking
tools: it is the most robust tool in our evaluation and often outperforms
competitors by several orders of magnitude on benchmarks from practical
applications.
\end{abstract}

\section{Introduction}


Language inclusion checking for Büchi automata (BAs) is a cornerstone problem in automata-theoretic verification, with applications to model checking~\cite{DBLP:conf/hvc/Vardi08}, monitoring~\cite{DBLP:conf/ictac/DiekertMW15}, and theorem proving~\cite{pecan,DBLP:journals/lmcs/HieronymiMOSSS24}.
The problem is \pspace-com\-plete~\cite{inclusionPSPACE}, but it remains one of the most computationally demanding tasks in verification practice.
As systems and specifications continue to grow in size and complexity, the scalability of inclusion checking has become a decisive factor in the practicality of automata-based reasoning frameworks.
For two $\omega$-regular languages $L_1$ and $L_2$, deciding whether $L_1 \subseteq L_2$ reduces to checking the emptiness of $L_1 \cap \overline{L_2}$, where $\overline{L_2}$ denotes the complement of $L_2$.
This exposes two key challenges:
\begin{inparaenum}[(i)]
  \item  constructing $\overline{L_2}$ compactly, and
  \item  efficiently exploring the (often enormous) state space of the product
    automaton for $L_1 \cap \overline{L_2}$.
\end{inparaenum}
Addressing either of these challenges translates into
performance improvements \mbox{in the considered domains.}

Existing tools primarily tackle these challenges by improving either the complementation of $L_2$ or the exploration for the product automaton.
For example, tools such as \bait~\cite{DoveriGPR21}, \forklift~\cite{forklift}, and \rabit~\cite{AbdullaCCHHMV11,AbdullaCCHHMV10} use Ramsey-based complementation~\cite{Buchi62,SistlaVW87}, while improving efficiency through state-space reduction techniques---such as subsumption and simulation in \rabit or well-quasiorders in \bait and \forklift.
In~contrast, \spot~\cite{Duret-LutzRCRAS22} relies on determinization-based complementation~\cite{Safra88,Piterman07,Redziejowski12}, enhanced by lightweight simulation and specialized emptiness checking procedures for Emerson–Lei acceptance conditions~\cite{BaierBD00S19}.
We observe that state-of-the-art inclusion checkers still largely rely on
classical Ramsey- or determinization-based complementation constructions,
which date back more than three decades~\cite{Buchi62,Safra88}, despite
substantial recent progress in \buchi complementation. This motivates an
efficient and extensible framework that brings modern complementation
techniques to practical inclusion checking.

\smallskip\noindent\textbf{Contributions.} 
We present \kofola~\cite{kofolaweb}, an efficient and modular tool for BA
complementation and language inclusion checking, which is based on a recent
decomposition-based complementation algorithm~\cite{HavlenaLLST23}. For an
inclusion query $\langof{\aut_1} \subseteq \langof{\aut_2}$, \kofola performs
a fine-grained structural analysis of the complemented automaton, namely
$\aut_2$, by classifying its maximal strongly connected components (SCCs) as:
\begin{inparaenum}[(i)]
  \item non-accepting components;
  \item inherently weak accepting components (IWACs), in which all cycles accept;
  \item initial almost deterministic accepting components (IADACs) and
        deterministic accepting components (DACs), which have only deterministic
        transitions (a~IADAC is a deterministic SCC reachable from initial
        states through deterministic SCCs only, allowing nondeterminism just on
        cycle-free parts between SCCs); and
  \item the remaining accepting SCCs, classified as nondeterministic accepting
        components (NACs).
\end{inparaenum}
Notably, IADACs are a \emph{newly identified component type} introduced in this paper. For each
component type, \kofola exploits its structural properties by applying a
dedicated complementation algorithm, combining the component constructions
in a synchronous manner.

For inclusion checking, \kofola constructs an $\omega$-automaton with a single
generalized Rabin pair acceptance condition,
$\Finof{\tacc 0} \land \Infof{\tacc 1} \land \cdots \land
\Infof{\tcacc 7 {k-1}}$, and performs \emph{on-the-fly} emptiness checking
during the modular construction. This avoids explicitly generating the full
product automaton. Although \spot can also check emptiness for generalized
Rabin conditions~\cite{BaierBD00S19}, it requires the whole automaton to be
constructed explicitly. To mitigate state-space explosion, \kofola implements
a new emptiness-checking algorithm for the single generalized Rabin pair
condition. The algorithm can be viewed as a variant of Couvreur's
algorithm~\cite{Couvreur99} extended with a $\Fin$ predicate. It is maximally
lazy: the search terminates as soon as the explored states suffice to decide
the result, helping to keep the generated part of the complement small, which is
crucial since complementation is the main source of blow-up in inclusion
checking.



We evaluated \kofola on complementation benchmarks drawn from diverse
literature sources, comparing it with state-of-the-art tools \spot,
\ranker~\cite{HavlenaLS22b}, and an initial implementation of
decomposition-based complementation~\cite{HavlenaLLST23}. \kofola outperforms
all competitors both in the number of solved instances and in the size of the
resulting complement automata.
We further evaluated \kofola on BA inclusion benchmarks arising from program
termination checking, verification of concurrent systems, model checking of
hyperproperties, and theorem proving. The results show that \kofola outperforms
the state-of-the-art inclusion checkers \bait, \forklift, \rabit, and \spot,
in many cases by several orders of magnitude, and \kofola solves the most instances. These
results establish \kofola as an efficient, extensible, and highly competitive
tool for BA complementation and inclusion checking.
\ifTR\else
More details are given in~\cite{techrep}.
\fi


%

\smallskip\noindent\textbf{Related work.} 
BA complementation algorithms can be broadly grouped into several families.
\emph{Ramsey-based}~\cite{Buchi62,SistlaVW87,BreuersLO12,DBLP:conf/fm/LiTTVZ21},
\emph{rank-based}~\cite{KupfermanV01,FriedgutKV06,Schewe09,KarmarkarC09,ChenHL19,HavlenaL21,HavlenaLS22a,HavlenaLS22b},
and \emph{slice-based}~\cite{slice_based,TsaiFVT14,AllredU18} procedures encode
information about word acceptance using different data structures.
\emph{Determinization-based} constructions first determinize the automaton and then
complement its acceptance condition, typically producing automata with acceptance
conditions that are more complex than \buchi~\cite{Safra88,Piterman07,Redziejowski12,DBLP:conf/atva/LodingP19,LiTFVZ22}.
\emph{Learning-based}
approaches~\cite{LiTZS18} infer a BA for the complement through membership
queries. More recent \emph{decomposition-based}
constructions~\cite{LiTFVZ22,HavlenaLLST23} partition the automaton according
to structural properties and apply specialized complementation procedures to
each block. \kofola builds on this decomposition-based approach, introduces the
new IADAC component type, and, as our experiments show, improves over previous
implementations~\cite{LiTFVZ22,HavlenaLLST23} as well as
rank- and determinization-based tools such as \ranker and \spot.

Several works also optimize BA language-inclusion testing, mostly within
Ramsey-based frameworks. \rabit~\cite{AbdullaCCHHMV11,AbdullaCCHHMV10} uses
simulation and subsumption, both within and across automata, to prune the
searched state space. Well-quasiorder-based approaches~\cite{DoveriGPR21,forklift}
are more symbolic: to decide $\langof{\aut_1} \subseteq \langof{\aut_2}$, they
perform a breadth-first fixpoint search over ultimately periodic words of
$\langof{\aut_1}$, using a well-quasiorder to discard words subsumed by already
processed ones.

\newcommand{\figKofolaArch}[0]{
    \ifTR
  \begin{wrapfigure}[11]{r}{50mm}
    \else
  \begin{wrapfigure}[12]{r}{50mm}
    \fi
    \ifTR
    \vspace*{-5mm}
    \fi
  \vspace*{-15mm}
  \hspace*{-3mm}
  \begin{minipage}{60mm}
  \resizebox{0.89\textwidth}{!}{
   \begin{tikzpicture}[
  auto,
  transform shape,
  node distance=20mm,
  >=stealth'
]
\tikzstyle{block}=[draw,rectangle,rounded corners=2mm]
\tikzstyle{fork}=[decorate, decoration={show path construction, lineto code={
        \draw(\tikzinputsegmentfirst)|-($(\tikzinputsegmentfirst)!.5!(\tikzinputsegmentlast)$)-|(\tikzinputsegmentlast);}}]

\filldraw node[block,fill=magenta!20,minimum width=55mm,minimum height=8mm] (frontend) at (0,0) {Frontend};
\filldraw node[block,fill=brown!20,minimum width=10mm,minimum height=42mm,left of=frontend,xshift=-20mm,yshift=-17mm] (spot) {\textsc{Spot}};
\filldraw node[block,fill=cyan!20,minimum width=55mm,minimum height=8mm,below of=frontend,yshift=9mm] (preproc) {Preprocessing and setup};
\filldraw node[block,fill=blue!20,minimum width=25mm,minimum height=14mm,below of=preproc,xshift=-15mm,yshift=0mm] (compl) {\begin{tabular}{c}Complement\\{\small (top level)}\end{tabular}};
\filldraw node[block,fill=green!20,minimum width=25mm,minimum height=14mm,below of=preproc,xshift=14mm,yshift=0mm] (incl) {\begin{tabular}{c}Inclusion\\{\small (on-the-fly }\\{\small emptiness check)}\end{tabular}};
\filldraw node[block,fill=black!10,minimum width=55mm,minimum height=8mm,below of=preproc,yshift=-15mm] (succ) {Successor generation};

\filldraw node[block,fill=red!20,minimum width=10mm,minimum height=8mm,below of=succ,xshift=-22.4mm,yshift=8mm] (succ1) {Succ${}_1$};
\filldraw node[block,fill=orange!20,minimum width=10mm,minimum height=8mm,below of=succ,xshift=-8mm,yshift=8mm] (succ2) {Succ${}_2$};
\node[minimum width=10mm,minimum height=8mm,below of=succ,xshift=7mm,yshift=8mm] (dots) {\ldots};
\filldraw node[block,fill=violet!20,minimum width=10mm,minimum height=8mm,below of=succ,xshift=22.4mm,yshift=8mm] (succn) {Succ${}_n$};

\node[above of=frontend,xshift=-10mm,yshift=-10mm] (a1) {$\aut_1$\texttt{.hoa}};
\node[above of=frontend,xshift=10mm,yshift=-10mm] (a2) {$\aut_2$\texttt{.hoa}?};

\node[below of=spot,yshift=-10mm] (out) {$\aut_1^c$\texttt{.hoa}};

\coordinate[below of=preproc,yshift=11.5mm] (fork);
\draw[->] (frontend.176) |- (spot.74);
\draw[->] (spot.70) |- (frontend.184);

\draw[->] (frontend) -- (preproc);

\draw[->] (preproc.176) |- (spot.57);
\draw[->] (spot.40) |- (preproc.184);


\draw (preproc) -- (fork);
\draw[->] (fork) -| (compl.90);
\draw[->] (fork) -| (incl.90);

\draw[->] (compl.221) -| (succ.170);
\draw[->] (succ.149) -| (compl.319);

\draw[->] (incl.221) -| (succ.35);
\draw[->] (succ.10) -| (incl.319);

\draw[->] (compl.180) |- (spot.291);
\draw[->] (spot.270) -- (out);

\draw[->] (succ.189) -| (succ1.120);
\draw[->] (succ1.60) -| (succ.191);
\draw[->] (succ.201) -| (succ2.120);
\draw[->] (succ2.60) -| (succ.216);
\draw[->] (succ.349) -| (succn.120);
\draw[->] (succn.60) -| (succ.351);

\draw[->] (a1) -- (frontend.158);
\draw[->] (a2) -- (frontend.22);

\end{tikzpicture}
  }
  \vspace{-2mm}
  \end{minipage}
  \caption{Architecture of \kofola}
  \label{fig:arch}
  \end{wrapfigure}
}

\section{Tool Architecture}\label{sec:arch}

\figKofolaArch  
\kofola is implemented in C++ and available at~\cite{kofolaweb} under the
GNU GPL~3.0 license. 
Its architecture is shown in \cref{fig:arch}.
The \emph{Frontend} uses \spot~\cite{Duret-LutzRCRAS22} to read 
input automata~$\aut_1$\texttt{.hoa} and, for inclusion checking only, $\aut_2$\texttt{.hoa}, both in the
\texttt{HOA}~\cite{BabiakBDKKM0S15} format.
Next, the automata enter \emph{Preprocessing and setup}, where
\begin{inparaenum}[(i)]
  \item  \spot reduces them using either \texttt{Low} or \texttt{High} reduction level, depending on their features and
  \item  \kofola then analyzes their structure and chooses partial complementation
    algorithms to be used. 
\end{inparaenum}
Then, the automaton/automata
are passed to top-level procedures for
\emph{Complement} (see \cref{sec:compl-alg}) or \emph{Inclusion}, which orchestrate the chosen partial complementation
algorithms (see \cref{sec:partial-compl}); for inclusion checking, this includes a top-level on-the-fly emptiness checking algorithm (see
\cref{sec:inclusion}).
The main workhorse of the procedures is \emph{Successor generation}
used in generating reachable states of the complement automaton, 
which calls the successor functions \emph{Succ}$_1$, \ldots, \emph{Succ}$_n$ of
the partial complementation algorithm and combines results.
For complementation, the resulting automaton is postprocessed by \spot
\mbox{(reduction level \texttt{Low}) and output as~$\aut_1^c$\texttt{.hoa}.}

\section{Preliminaries}\label{sec:label}

We assume familiarity with standard notions on $\omega$-automata; see, e.g.,~\cite{EsparzaB23}.
We only introduce definitions specific to this paper.
Fix a~finite non-empty alphabet~$\alphabet$, and let~$\omega$ be the first
infinite ordinal.
An~(infinite) word is a sequence $\word = \wordof 0
\wordof 1 \dots$ over~$\alphabet$;
the set of all infinite words over~$\alphabet$
is denoted by $\alphabet^\omega$.

A~(nondeterministic transition-based) \emph{simple generalized Rabin automaton}
(SGRA)
over~$\Sigma$ is a~quintuple $\aut = (Q, \trans, I, \colours, \colouring)$
where $Q$ is a~finite set of \emph{states}, $\trans \subseteq Q \times \Sigma \times
Q$ is a~set of \emph{transitions} (denoted as $q \ltr a r$), $I \subseteq Q$ is the set of
\emph{initial} states, $\colours = \{\tacc 0, \ldots, \tcacc 7 {k-1}\}$ is a~set of~$k$
\emph{colours}, and $\colouring\colon \delta \to 2^\colours$ is a~colouring
function on transitions.
For a~set of states~$C \subseteq Q$, we write~$\restrof{\trans}{C}$ to denote the set of transitions $\{q \ltr a r \in \trans \mid q,r \in C\}$.
We call~$\aut$ \emph{deterministic} if~$|I| = 1$ and 
for every~$a\in \alphabet$ and~$q \in Q$ it holds that $|\{p \mid \in (q, a, p) \in \trans\}| \leq 1$.
When~$\tacc 0$ has no occurrence in~$\colouring$, $\aut$ is a \emph{generalized \buchi automaton} (GBA).
When a~GBA has $\colours = \{\tacc 0, \tacc 1\}$, it is a~\emph{\buchi automaton} (BA).
We~denote a~BA as $\aut = (Q, \trans, I, \acctrans)$ where
$\acctrans = \{\tau \in \trans \mid \colouringof \tau = \{\tacc 1\}\}$.
A~\emph{run}
of~$\aut$ from~$q \in Q$ on $\word\in\alphabet^\omega$ is an infinite sequence $\rho\colon
\omega \to Q$ that $\rho_0 = q$ and
$\rho_i \ltr{\wordof i}\rho_{i+1} \in \trans$ for all $i \geq 0$. 
Let $\inftof \rho$ denote the transitions occurring in~$\rho$ infinitely often.
The run~$\rho$ is called \emph{accepting} iff $\bigcup \{\colouringof \tau \mid \tau
\in \inftof \rho\} = \colours \setminus \{\tacc 0\}$.
Equivalently, in the Emerson--Lei formalism~\cite{EmersonL87,BabiakBDKKM0S15},
SGRAs have acceptance condition $\Finof{\tacc 0} \land \Infof{\tacc 1} \land
\ldots \land \Infof{\tcacc 7 {k-1}}$.
A word $\word$ is accepted by $\aut$ from a state $q$ if there exists an
accepting run of $\aut$ from $q$ on $\word$. We write $\langof{\aut}$ for the
language accepted from the initial states $I$. A state $q$ is \emph{useless}
if no word is accepted from $q$.

For a BA $\aut$, a non-empty set $C\subseteq Q$ is a \emph{strongly connected
component} (SCC) if all states in~$C$ can reach each other and~$C$ is a~maximal such set.
An SCC is \emph{trivial} if it is a singleton without a
self-loop, and \emph{non-trivial} otherwise.
An~SCC~$C$ is
\emph{accepting} if it contains at least one $\tacc 1$-transition and
\emph{inherently weak} iff either
\begin{inparaenum}[(i)]
  \item  every run in~$\aut$ restricted to~$C$ is accepting, or 
  \item  no run in~$\aut$ restricted to~$C$ is accepting.
\end{inparaenum}
An SCC~$C$ is \emph{deterministic} iff the BA $(C, \restrof{\trans}{C},
\{q\}, \emptyset)$ for any~$q \in C$ (obtained from~$\aut$ by
removing transitions outside~$C$) is deterministic, and it
is \emph{initial deterministic} if the BA obtained from
the BA $(Q, \trans, I, \restrof{\acctrans} C)$ (i.e., $\aut$~with accepting
transitions outside~$C$ removed) by removing useless states is
deterministic.
In particular, a~non-accepting SCC is not initial deterministic since the resulting BA has no initial states.
We also generalize the notion of initial deterministic components to
\emph{initial almost deterministic components}, which are components~$C$ such
that in the BA obtained from the BA $(Q, \trans, I, \restrof{\acctrans} C)$
by removing useless states, it holds that for any two transitions from the same
state over the same symbol $q \ltr a p$ and $q \ltr a r$, either $p=r$ or the
states~$p$ and~$r$ are not in the same SCC as the state~$q$.
Hence, such components are reachable from initial states on a path going through deterministic
components with nondeterminism allowed only outside non-trivial SCCs.
We use the following categorization of accepting SCCs of a~BA:
\begin{inparaenum}[(i)]
  \item \textbf{IADACs} are initial almost deterministic accepting SCCs,
  \item \textbf{IWACs} are inherently weak accepting SCCs that are not IADACs,
  \item \textbf{DACs} are deterministic accepting SCCs not covered by the above, and
  \item \textbf{NACs} are the remaining (nondeterministic) accepting SCCs.
\end{inparaenum}
A~BA~$\aut$ is called an \emph{elevator automaton}~\cite{HavlenaLS22a} if it contains no NAC.
We assume that for all transitions $\tau = q \ltr a r \in \delta$ not
within a~single SCC, we have $\colouringof \tau = \emptyset$ (this is without
loss of generality since no run can loop over such transitions).
A~\emph{partition block} $P \subseteq Q$ of~$\aut$ is a~nonempty union of~$\aut$'s
accepting SCCs, and a~\emph{partitioning} of~$\aut$ is a~sequence $P_1, \ldots,
P_n$ of pairwise disjoint \mbox{partition blocks of~$\aut$ that contains all accepting SCCs
of~$\aut$.}

\section{Decomposition-based \buchi Complementation}\label{sec:complementation}

We fix a~BA $\aut = (Q, \trans, I, \acctrans)$ for the rest of this section.
We first give a~brief overview of the complementation
framework proposed in~\cite{HavlenaLLST23} and then present our choices of partial
complementation algorithms for each type of partition blocks.
We also extend the framework of~\cite{HavlenaLLST23} with a~new partial
complementation algorithm for IADACs.


\subsection{Top-level Modular Complementation Algorithm}\label{sec:compl-alg}

Inspired by~\cite{LiTFVZ22}, a modular methodology was proposed in \cite{HavlenaLLST23} for complementing~BAs.
In a nutshell, the modular complementation algorithm developed in~\cite{HavlenaLLST23} first classifies the SCCs in~$\aut$ into several partition blocks according to their structural properties, then performs
complementation for each of the partition blocks independently, while
synchronizing the complementation algorithms for all partition blocks in each step.
The intuition for complementing SCCs separately is that every run of~$\aut$
eventually gets trapped in one SCC, so we only need to consider the runs staying in the given SCC.
For those runs that exit the given SCC at some point, we handle them as discontinued runs since they must eventually stay in another SCC. 
The complement automaton~$\aut^c$ accepts a~word over which all runs of~$\aut$ are rejecting, so we only need to make sure that all runs in \emph{accepting} SCCs over the word do not visit accepting transitions infinitely often.
That is, we check only the runs in IADACs, IWACs, DACs, and NACs. 
Runs in rejecting SCCs do not require checking acceptance, so we just track them using traditional subset construction.
We choose to decompose $\aut$ into partition blocks where
one partition block contains all IADACs,
one contains all IWACs,
one contains all DACs,
and each NAC has
a~separate partition block. However, thanks to the modularity of the complementation framework,
we can decompose $\aut$ in other ways and even use different partial complementation algorithms for partition blocks of the same type.
The modular complementation algorithm uses a partial algorithm $\alg^P$ for each partition block $P$.
Each partial algoritm $\alg^P$ specifies the following:
\begin{inparaenum}[(i)]
  \item the type of macrostates produced by the algorithm, 
  \item the set of initial macrostates,
  \item a~function $\GetSucc$ returning the successors of a macrostate, possibly together with a transition color, and
  \item the acceptance condition for \mbox{the partial algorithm.}
\end{inparaenum}%
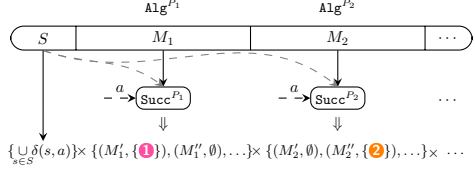
\begin{wrapfigure}[7]{r}{0.5\textwidth}
	\centering
	\vspace{-11.5mm}
	\usetikzlibrary{calc}

\begin{tikzpicture}[
  auto,
  transform shape,
  node distance=20mm,
  >=stealth',
  scale=0.60
]

\newcommand{\wS}{12mm}
\newcommand{\wM}{36mm}
\newcommand{\wDots}{8mm}

\node[rectangle split, rectangle split horizontal, rectangle split parts=4,
  rectangle split draw splits=true, draw, rounded corners=5pt,
  align=center] (pred) {\makebox[\wS][c]{$S$}
  \nodepart{two} \makebox[\wM][c]{$M_1$}
  \nodepart{three} \makebox[\wM][c]{$M_2$}
  \nodepart{four} \makebox[\wDots][c]{$\cdots$}};

\foreach \i/\prt in {1/two,2/three} {
  \node[font=\small, anchor=south] at ([yshift=1mm]pred.\prt\space north) {$\alg^{P_{\i}}$};
  \node[draw, rounded corners=3pt, font=\small, inner sep=3pt, anchor=north] (alg\i)
    at ([yshift=-8mm]pred.\prt\space south) {$\GetSucc^{P_{\i}}$};

  \draw[-stealth, dashed] ([xshift=-7mm]alg\i.west) -- node[above] {$a$} (alg\i.west);

  \ifthenelse{\equal{\i}{1}}{
  \node[font=\small, inner sep=2pt, anchor=north,yshift=0.4mm] (mods\i)
    at ([yshift=-7mm]alg\i.south) {${}\times\{(M_{\i}', \{\tacc 1\}), (M_{\i}'',\emptyset), \ldots\}$};
  }{
  \ifthenelse{\equal{\i}{2}}{
  \node[font=\small, inner sep=2pt, anchor=north,yshift=0.4mm] (mods\i)
    at ([yshift=-7mm]alg\i.south) {${}\times\{(M_{\i}', \emptyset), (M_{\i}'', \{\tacc 2\}), \ldots\}$};
  }}

  \node at ($(alg\i.south)!0.5!(mods\i.north)$) {$\Downarrow$};

  \coordinate (predSplitSouth\i) at (pred.\prt\space south);
  \draw[-stealth] (predSplitSouth\i) -- (alg\i.north);
}

\coordinate (predFirstSplitSouth) at (pred.one\space south);

\foreach \j in {1,2} {
  \draw[-stealth, dashed, gray] (predFirstSplitSouth) to[out=340,in=140] (alg\j.north);
}

\node[inner sep=2pt, anchor=north] (alg3) at ([yshift=-9.5mm]pred.four\space south) {$\cdots$};

\node[font=\small, inner sep=2pt, anchor=north] (mods3)
    at ([yshift=-8.3mm,xshift=0.5mm]alg3.south) {${}\times\hspace{0.1cm}\cdots\hspace{0.3cm}$};
\node[font=\small, inner sep=2pt, anchor=north] (deltas)
    at ([yshift=-20mm, xshift=-0mm]pred.one\space south) {$\{\!\underset{s\in S}{\cup}\!\delta(s,a) \}$};
\draw[-stealth] (pred.one\space south) -- (deltas.north);



\end{tikzpicture}
	\vspace{-7mm}
	\caption{Visualization of macrostate's successor computation over symbol $a$.}
	\label{fig:succ}
\end{wrapfigure}%
The input of the top-level algorithm are, except the automaton itself, also partitions $P_1,\ldots,P_k$ of the input automaton and the 
partial algorithms $\alg^{P_1},\ldots, \alg^{P_k}$. The macrostates of the top-level algoritm are then of the form $(S, M_1, \dots, M_k)$, where 
$S$ is the set of reachable states and $M_1,\ldots, M_k$ are macrostates of the partial algorithms. The successors of a~macrostate 
are built from the results of partial algorithms' successor functions (and
colors from partial successor functions are gathered on the transition)
as shown in \cref{fig:succ}. The acceptance condition of~$\aut^c$ is given as 
a~conjunction of acceptance conditions of the partial algorithms.

\subsection{Partial Complementation Algorithms}\label{sec:partial-compl}
Next, we briefly describe partial algorithms to complement partition blocks with
IWACs, DACs, and NACs. Details and formal definitions can be found in~\cite{HavlenaLLST23}.
Finally, we also present a new partial complementation algorithm for partition blocks with IADACs only.
The sole source of the $\Finof {\tacc 0}$ acceptance condition comes from IADACs, 
making the resulting automaton an SGRA.

\smallskip
\noindent\textbf{IWACs.}
The algorithm for the partition block $\iw$ that contains all IWACs exploits the classical \textit{Miyano-Hayashi} construction~\cite{MIYANO1984321}, generating
the acceptance condition~$\Infof{\tacc 1}$.

\smallskip
\noindent\textbf{DACs.}
For the partition block with DACs, we use a modification of the NCSB
algorithm~\cite{BlahoudekDS20,ChenHLLTTZ18,HavlenaLS22b}, which generates the acceptance
condition $\Infof{\tacc 2}$.

\smallskip
\noindent\textbf{NACs.}
For partition blocks containing a~NAC~$\nac$, \kofola supports two algorithms:
\begin{inparaenum}[(i)]
  \item  determinization-based (based on the
    Safra-Piterman's~\cite{Safra88,Piterman07} algorithm, specifically its
    version
    from~\cite{LiTFVZ22}), which translates~$\nac$ into a~deterministic
    parity automaton, and
  \item  a~slice-based approach from~\cite{AllredU18} (this approach generates
    for a~given NAC the acceptance condition $\Infof {\tcacc 7 {k}}$ with $k > 2$; we
    note that \tacc 0, \tacc 1, and \tacc 2 are reserved for the procedures for
    IADACs, IWACs, and DACs respectively).
\end{inparaenum}
%
Each of the algorithms is used in a~different setting.
The determinization-based algorithm is used in \emph{complementation}, since it
usually produces smaller results.
On the other hand, the slice-based algorithm is used in \emph{inclusion
checking}, since it is more suitable for on-the-fly product emptiness~test.

\smallskip
\noindent\textbf{IADACs.}
Let us now generalize the partial algorithm for complementing initial
deterministic components (IDCs) from~\cite{HavlenaLLST23} to IADACs.
We show the construction on a BA consisting of IADACs only. 
Putting the construction into the formal framework of~\cite{HavlenaLLST23} is 
straightforward. Given a BA $\aut = (Q, \trans, I, \acctrans)$ consisting of 
IADACs only, the complement SGRA with the acceptance condition $\Finof{\tacc 0}$ 
is given as $\but = (2^Q, \Delta, \{I\}, \{\tacc 0\}, \colouring)$ where 
%
$\Delta = \{R \ltr a S \mid \bigcup_{r \in R} \trans(r, a) = S\}$
and
for a transition $t = R \ltr a {S}$ the coloring function $\colouring$ 
is given as $\colouring(t) = \{ \tacc 0 \}$ if $ \bigcup_{r \in R} \acctrans(r, a) \neq \emptyset$ 
otherwise $\colouring(t) = \emptyset$.

Intuitively, the algorithm determinizes the IADACs using a subset construction
and emits a~rejecting (due to the $\Fin$ condition) mark~$\tacc 0$ whenever
some tracked run sees an accepting transition.
The key argument for correctness is that the number of runs in
IADACs is bounded (the directed acyclic graph representing the run---called the
\emph{run DAG}---has a~finite number of branches), 
due to IADACs having non-deterministic branching only outside of (non-trivial) SCCs.
(On the other hand, observe that a~DAC having non-deterministic branching between another DAC and itself would generate unboundedly many runs.)
%
%
%
\begin{theorem}
	$\lang(\aut) = \Sigma^\omega \setminus \lang(\but)$.
\end{theorem}
\begin{proof}[Sketch]
Since all IADAC runs are deterministic, deciding whether $w$ is accepted
amounts to checking whether the runs on $w$ visit $\acctrans$-transitions
only finitely often. By construction, $w$ is not accepted in the IADAC iff
every run of~$\aut$ over~$w$ emits only finitely many $\tacc0$'s;
equivalently, the corresponding run in $\but$ emits $\tacc0$ only finitely
often, as captured by the acceptance condition $\Finof{\tacc 0}$.
\qed\end{proof}

\newcommand{\ass}{\leftarrow}
\newcommand{\addtoset}[2]{{#1} \ass {#1} \cup \{#2\}}
\newcommand{\removefromset}[2]{#1 \ass #1\setminus \{#2\}}
\newcommand{\discovered}{D}
\newcommand{\isempty}{\mathit{Empty}} 
\newcommand{\stack}{S}
\newcommand{\push}{\mathtt{push}}
\newcommand{\falseval}{\mathit{false}}
\newcommand{\trueval}{\mathit{true}}
\newcommand{\pop}{\mathtt{pop}}
\newcommand{\stacktop}{\mathtt{top}}
\newcommand{\entry}{\mathit{Entry}}
\newcommand{\explore}{\mathtt{Explore}}
\newcommand{\explored}{\mathit{Explored}}
\newcommand{\source}{\mathtt{src}}
\newcommand{\target}{\mathtt{tgt}}
\newcommand{\outgoing}{\mathtt{outgoing}}
\newcommand{\algEmptinessTwoCol}[0]{%
	\begin{algorithm}[t]
		\caption{On-the-fly emptiness test for SGRAs}\label{alg:topLevel}
		\begin{multicols}{2}
			\begin{algorithmic}[1]
				\Input an SGRA $\aut = (Q, \trans, I, \{\tacc 0, \ldots, \tcacc 7
				{k-1}\}, \colouring)$
				\Output if $\langof \aut = \emptyset$ then $\trueval$ else $\falseval$
				\Data sets of transitions $\entry$, $\explored$
				\DataEmpty a stack $\stack$ of sets of transitions
				
				\LState $\entry \gets \outgoing(I)$\label{line:initentry}
				\LState $\explored \gets \emptyset$
				\While{$\exists t \in \entry {\setminus} \explored$}\label{line:mainwhile}
				\LState $\explore(t)$
				\EndWhile\popindent%
				\LState \textbf{exit algorithm with} $\trueval$
				
				\Procedure{$\explore$}{$t$}\label{line:exploreproc}
				\LState $\explored \gets \explored \cup \{t\}$\label{line:markexplored}
				\If{$\tacc 0 \in \colouring(t)$}\label{line:badedgetest}
				\LState $\entry \gets \entry \cup \{\outgoing(\target(t))\}$\label{line:addtoentry}
				\Else\label{line:normalexplore}
				\LState $\stack.\push(\{t\})$
				\LState Merge the stack $\stack$ from the lowest occurrence of $\target(t)$ up to top\label{line:merging}
				\If{$\{\tacc 1, \ldots, \tcacc 7 {k-1}\} \subseteq \colouring(\stack.\pop)$}\label{line:buchitest}
				\LState \textbf{exit algorithm with} $\falseval$
				\EndIf\popindent%
				\While{$\exists t' \in \outgoing(\target(\stack.\pop)){\setminus} E$}
				\LState $\explore(t')$
				\EndWhile\popindent %
				\LState $\stack.\pop()$\label{line:popterminal}
				\EndIf\popindent%
				\EndProcedure\popindent%
			\end{algorithmic}
		\end{multicols}
	\end{algorithm}%
}

\newcommand{\algEmptinessOneCol}[0]{
\begin{figure}[t]%
\SetKw{KwMyIf}{if}
\SetKw{KwOr}{or}
\SetKw{KwThen}{then}
\SetKw{KwExit}{exit program}
\SetKw{KwExitAndReturn}{exit program and return}
\SetKwProg{Procedure}{Procedure}{:}{}
\SetInd{0.3em}{1.2em}
\resizebox{\textwidth}{!}{
\begin{minipage}{14cm}
  \begin{algorithm}[H]
    \caption{On-the-fly emptiness test for SGRAs}
    \label{alg:topLevel}
    \KwIn{an SGRA $\aut = (Q, \trans, I, \colours = \{\tacc 0, \tacc 1, \ldots, \tcacc 7
      {k-1}\}, \colouring)$}
    \KwOut{$\trueval$ iff $\langof \aut = \emptyset$, $\falseval$ otherwise}
    \KwData{sets of transitions $\entry$, $\explored$,
      a stack $\stack$ of sets of transitions}
    %
    \BlankLine
    $\entry\ass\outgoing(I)$;\,$\explored \ass \emptyset$;\,$S\ass ()$\;\label{line:initentry}
    \lWhile{$\exists t \in \entry\setminus\explored$}
    {\label{line:mainwhile}
      $\explore(t)$
    }
    \KwRet $\trueval$\;
    \BlankLine
    \Procedure{$\explore(t)$}{\label{line:exploreproc}
      $\addtoset{\explored}{t}$\;\label{line:markexplored}
      \lIf{$ \tacc 0 \in \colouring(t)$} {\label{line:badedgetest}
        $\addtoset{\entry}{\outgoing(\target(t))}$\label{line:addtoentry}
      }
      \Else {\label{line:normalexplore}
        $\stack.\push(\{t\})$\;
        Merge the stack  $\stack$ from the lowest occurrence of $\target(t)$ up to top\;\label{line:merging}
        \lIf{$\{\tacc 1, \ldots, \tcacc 7 {k-1}\} \subseteq \colouring(\stack.\stacktop)$}{\label{line:buchitest}
          \KwExitAndReturn $\falseval$
        }
        \lWhile{$\exists t' \in \outgoing( \target (\stack.\stacktop))\setminus\explored$ } {
          $\explore(t')$
        }
        $\stack.\pop()$\;\label{line:popterminal}
      }
    }
  \end{algorithm}
\end{minipage}
}
\end{figure}
}

\section{Efficient \buchi Inclusion Checking}\label{sec:inclusion}


On the top level, our algorithm for testing the inclusion $\langof{\aut_1}
\subseteq \langof{\aut_2}$ for BAs~$\aut_1$ and~$\aut_2$
checks emptiness of the SGRA~$\aut$ constructed as the intersection of~$\aut_1$
and~the SGRA~$\aut_2^c$ for the complement of~$\aut_2$.
Note that~$\aut$ can be explored by performing the
product construction of~$\aut_1$ with~$\aut_2^c$ on the fly, extending the
colours of $\aut_2^c$ with one more colour that represents the acceptance of~$\aut_1$.
%
%
Since the biggest asymptotic state space explosion comes from generating the
complement~$\aut_2^c$, our goal is to have an emptiness checking
algorithm that is maximally lazy in the sense that it stops generating new
states of~$\aut_2^c$ as soon as the already generated states
are sufficient for establishing the emptiness of~$\aut$.
This precludes the use of algorithms such as~\cite{BaierBD00S19}, which
require the whole automaton to identify its SCCs.
To achieve this goal, below, we describe a~new language emptiness algorithm for SGRAs
that satisfies the laziness condition.




\cref{alg:topLevel} performs on-the-fly emptiness test for SGRAs.
It uses the notation $\source(t)$/$\target(t)$ for the source/target state of a~transition $t$, respectively,  
$\outgoing(q)$ to denote the outgoing transitions of a~state~$q$, and lifts the colouring function $\colouring$ to sets in the usual manner. 
The algorithm is an augmentation of the on-the-fly emptiness check for GBAs by
Couvreur~\cite{Couvreur99}, which may also be seen as a variation of Tarjan's algorithm~\cite{Tarjan} for exploring SCCs.
To avoid technicalities, \cref{alg:topLevel} is based on its high-level presentation.
We emphasize that although this algorithm does not immediately resemble Tarjan's algorithm, the merged sets of transitions might efficiently be implemented similarly as Tarjan's SCCs, using a~technique resembling Tarjan's stack with its low-links.

\algEmptinessOneCol 

The algorithm explores the automaton in a depth-first manner 
from~$I$. 
The stack $\stack$ stores the currently explored path. 
When a back-edge transition $t$ (a~transition to a state already on the stack) is detected, 
the top part of the stack from the element including the target of~$t$ is merged and the united set of transitions is placed on the top of the stack.
The set represents a (non-maximal) SCC. 
The stack thus becomes a~sequence of sets of transitions.
\lnref{line:merging} specifically transforms the stack $\stack$, a sequence of sets of transitions $T_1,\ldots,T_n$ (the top being the rightmost element), into $T_1,\ldots,T_{\ell-1},\bigcup_{i=\ell}^n T_i$ where $\ell$ is the lowest index such that $T_\ell$ contains a~transition $t'$ with $\source(t') = \target(t)$, or it can be also $\target(t') = \target(t)$ if $T_\ell$ has a cycle. 
Next, Couvreur's algorithm concludes non-emptiness when it detects a~GBA cycle, i.e., a~merged set of transitions on top of the stack that contains all colours.
\cref{alg:topLevel}, on the other hand, needs to find a~\emph{simple generalized Rabin accepting cycle}, i.e., one that satisfies the GBA acceptance and does not contain~$\tacc 0$. 
In other words, it is a~GBA accepting cycle with respect to the set $\{\tacc 1, \ldots, \tcacc 7 {k-1}\}$ in the automaton $\aut \setminus \tacc 0$
(the automaton~$\aut$ with $\tacc 0$-transitions removed)
that is still reachable from $I$ in~$\aut$. 
To detect such cycles, we run a variant of Couvreur's algorithm as a sub-procedure, each time in a part of $\aut\setminus \tacc 0$ discovered to be reachable in~$\aut$.
Each run of Couvreur's algorithm is used
\begin{inparaenum}[(i)]
  \item to detect SGRA accepting cycles in a part of $\aut\setminus \tacc 0$
  rooted at some transition~$t$; and
  \item to discover new parts of $\aut\setminus \tacc 0$ reached from the
  explored part by $\tacc 0$-transitions. 
\end{inparaenum}
This is iterated until a~GBA accepting cycle is found or there are no more unexplored reachable parts of $\aut\setminus \tacc 0$. 
%

\cref{alg:topLevel} uses the set $\entry$ of \emph{entry transitions} as a~data
structure initialized by the outgoing transitions of $I$
(\lnref{line:initentry}), 
and applies Couvreur's algorithm to explore parts of $\aut\setminus \tacc 0$
from all transitions of $\entry$, while the iterations of the Couvreur's algorithm
also keep adding transitions to $\entry$ (due to the \textbf{while} loop on
\lnref{line:mainwhile}). 
We modify the internals of Couvreur's algorithm, presented essentially as the
procedure $\explore$ on \lnref{line:exploreproc}, in one point:
When exploring a~$\tacc 0$-transition, 
the normal depth-first exploration (the \textbf{else} branch on \lnref{line:normalexplore}) is bypassed, the transition is only marked as explored (\lnref{line:markexplored}), and the transitions continuing from its target are put to $\entry$ as the entry transitions of newly discovered parts of $\aut\setminus \tacc 0$ (\lnref{line:addtoentry}).   
\begin{theorem}
	\Cref{alg:topLevel} terminates and returns $\falseval$ iff $\langof{\aut} \neq \emptyset$. 
\end{theorem}
\begin{proof}[Sketch]
\Cref{alg:topLevel} terminates since each visited transition is stored in $\explored$ and never revisited, and $\aut$ has finitely many transitions.
The proof follows the proof of Couvreur's algorithm. We show that $\aut$ has an accepting run iff \Cref{alg:topLevel} returns $\falseval$. First, $\langof \aut \neq \emptyset$ iff $\aut\setminus \tacc 0$ contains a reachable cycle visiting all colours $\tacc 1,\ldots,\tcacc 7 {k-1}$. Whenever a transition $t$ has colour $\tacc 0$, the algorithm stops the current DFS branch and schedules $\target(t)$ as a new entry point (Line~6). Since all explored transitions are reachable from an initial state, every state in $\entry$ is reachable; so the search effectively explores $\aut\setminus\tacc 0$.
Each newly visited transition is pushed onto $\stack$ (Line~8). If it closes a cycle, namely if $\target(t)$ equals $\source(t')$ or $\target(t')$ for some explored $t'$ on $\stack$, the colours on that cycle are collected and stored in $\stack.\stacktop$ (Line~9).
If the algorithm returns $\falseval$, then $\stack.\stacktop$ holds all colours $\tacc 1,\ldots,\tcacc 7 {k-1}$ (Line~10). So a reachable cycle in $\aut\setminus\tacc 0$ visiting all non-$\tacc 0$ colours has been found, satisfying the accepting condition; therefore $\aut$ has an accepting run.
Conversely, if $\aut$ has an accepting run, then $\aut\setminus\tacc 0$ contains a reachable cycle visiting all non-$\tacc 0$ colours. Once the search reaches the cycle, it coincides with Couvreur's algorithm on $\aut\setminus\tacc 0$ and must discover it. So Line~10 is eventually triggered, and \Cref{alg:topLevel} returns $\falseval$.
\qed
\end{proof}
\newcommand{\figScatterCompl}[0]{
	\begin{figure}[t]
		\centering
		\begin{subfigure}{0.32\textwidth}
			\centering
			\includegraphics[width=3.8cm,alt={A log-log scatter plot comparing the state counts of complement automata generated by Kofola (x-axis) and Spot (y-axis), ranging from approximately 3 to 2×10⁵. Dashed lines mark the out-of-resources limit on both axes. A diagonal dashed line represents equal performance; points above it indicate Spot produces larger automata than Kofola. Data points are coloured by benchmark source: purple (SOBC), yellow (Seminator), red (AutomizerC), light blue (LDBA4LTL), dark blue (AutoHyperC), green (PecanC), and orange (S1S).
The large majority of purple and yellow points cluster densely above the diagonal at low-to-mid values (roughly 3–10³), indicating Kofola frequently produces smaller automata than Spot on SOBC and Seminator benchmarks. At higher state counts (10³–10⁵), red (AutomizerC) and green (PecanC) points appear prominently along or near the diagonal, with several red outliers above it, suggesting Spot occasionally yields larger automata on those benchmarks as well. A light blue (LDBA4LTL) point sits notably high on the y-axis near the out-of-resources limit, while the corresponding Kofola value is comparatively small. Marginal tick marks along the top and right edges show the per-axis distribution of values, coloured by benchmark source.}]{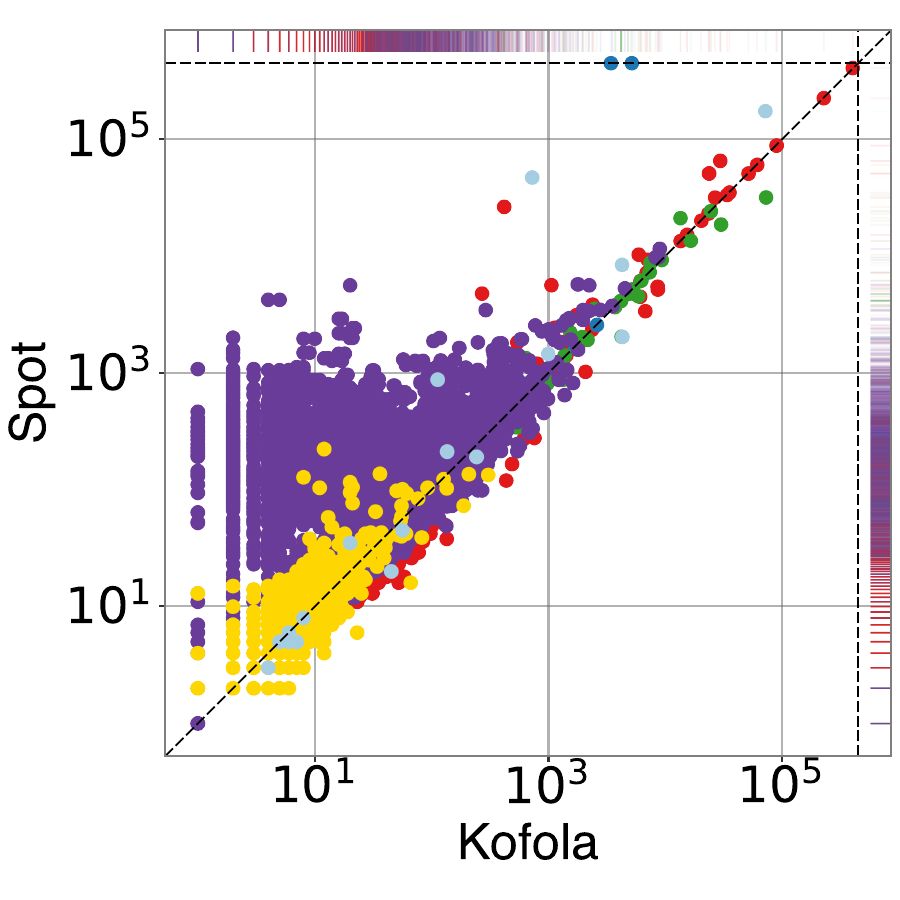}
		\end{subfigure}
		\begin{subfigure}{0.32\textwidth}
			\centering
			\includegraphics[width=3.8cm,alt={A log-log scatter plot comparing the state counts of complement automata generated by Kofola (x-axis) and KofolaOld (y-axis), ranging from approximately 3 to 2×10⁵. Dashed lines mark the out-of-resources limit on both axes. A diagonal dashed line represents equal performance; points above it indicate KofolaOld produces larger automata than Kofola. Data points are coloured by benchmark source: purple (SOBC), yellow (Seminator), red (AutomizerC), light blue (LDBA4LTL), dark blue (AutoHyperC), green (PecanC), and orange (S1S).
Purple and yellow points form a dense cluster at low-to-mid values (roughly 3–10³) sitting on or slightly above the diagonal, indicating the two versions perform very similarly on SOBC and Seminator benchmarks, with KofolaOld marginally larger in many cases. Green (PecanC) points show the most divergence — several appear well above the diagonal, including one outlier near (10, 4×10⁵), indicating KofolaOld produces dramatically larger automata than Kofola on some PecanC instances. Red (AutomizerC) points scatter along the diagonal at higher state counts, with a number hitting the out-of-resources limit on the KofolaOld axis while Kofola remains within bounds. Light blue (LDBA4LTL) and dark blue (AutoHyperC) points are sparse but generally near the diagonal. Marginal tick marks along the top and right edges show the per-axis distribution of values, coloured by benchmark source.}]{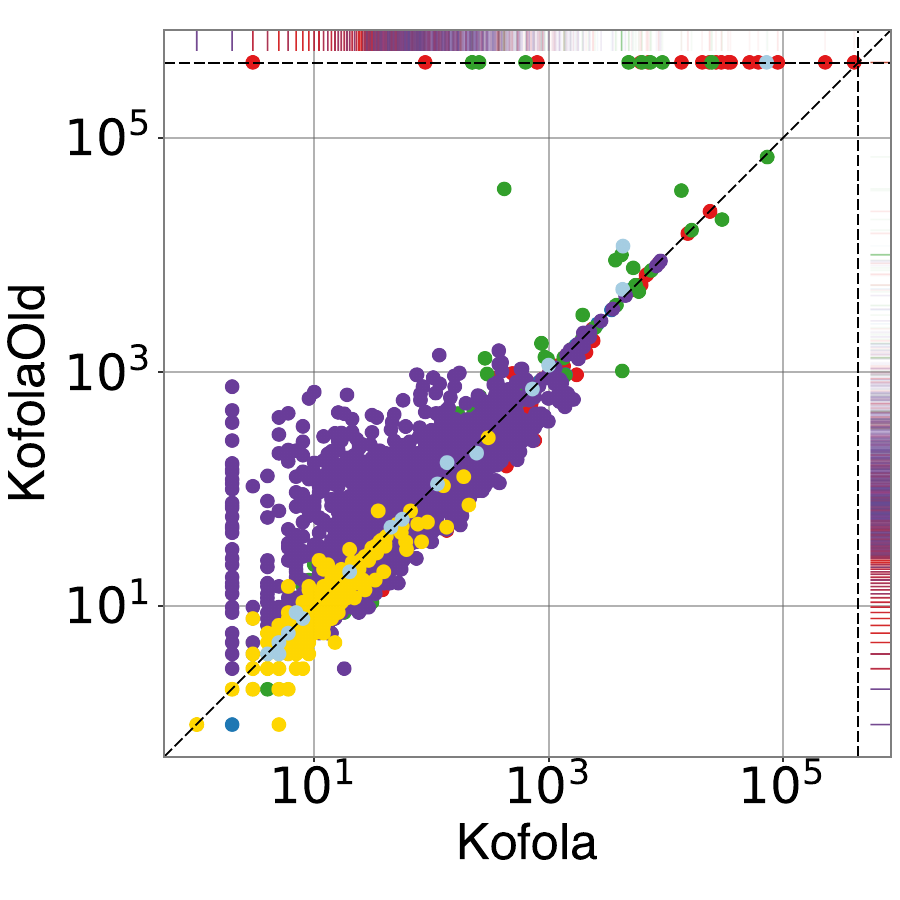}
		\end{subfigure}
		\begin{subfigure}{0.32\textwidth}
			\centering
			\includegraphics[width=3.8cm,alt={A log-log scatter plot comparing the state counts of complement automata generated by Kofola (x-axis) and Ranker (y-axis), ranging from approximately 3 to 2×10⁵. Dashed lines mark the out-of-resources limit on both axes. A diagonal dashed line represents equal performance; points above it indicate Ranker produces larger automata than Kofola. Data points are coloured by benchmark source: purple (SOBC), yellow (Seminator), red (AutomizerC), light blue (LDBA4LTL), dark blue (AutoHyperC), green (PecanC), and orange (S1S).
This plot shows the most pronounced advantage for Kofola of the three comparisons. A large number of purple and yellow points sit clearly above the diagonal across all value ranges, and a dense vertical streak of purple points at low Kofola values (around 3–5) with Ranker values spanning 10²–10⁴ indicates many cases where Ranker produces far larger automata than Kofola. Red (AutomizerC) points are scattered broadly above the diagonal, particularly at mid-to-high state counts. Light blue (LDBA4LTL) points also appear above the diagonal at mid-range values. A notable single teal (AutoHyperC) point sits well below the diagonal near (10, 3), indicating one instance where Ranker produces a much smaller automaton. The top marginal strip is heavily populated, reflecting that Ranker frequently hits the out-of-resources limit while Kofola does not. Marginal tick marks along the top and right edges show the per-axis distribution of values, coloured by benchmark source.}]{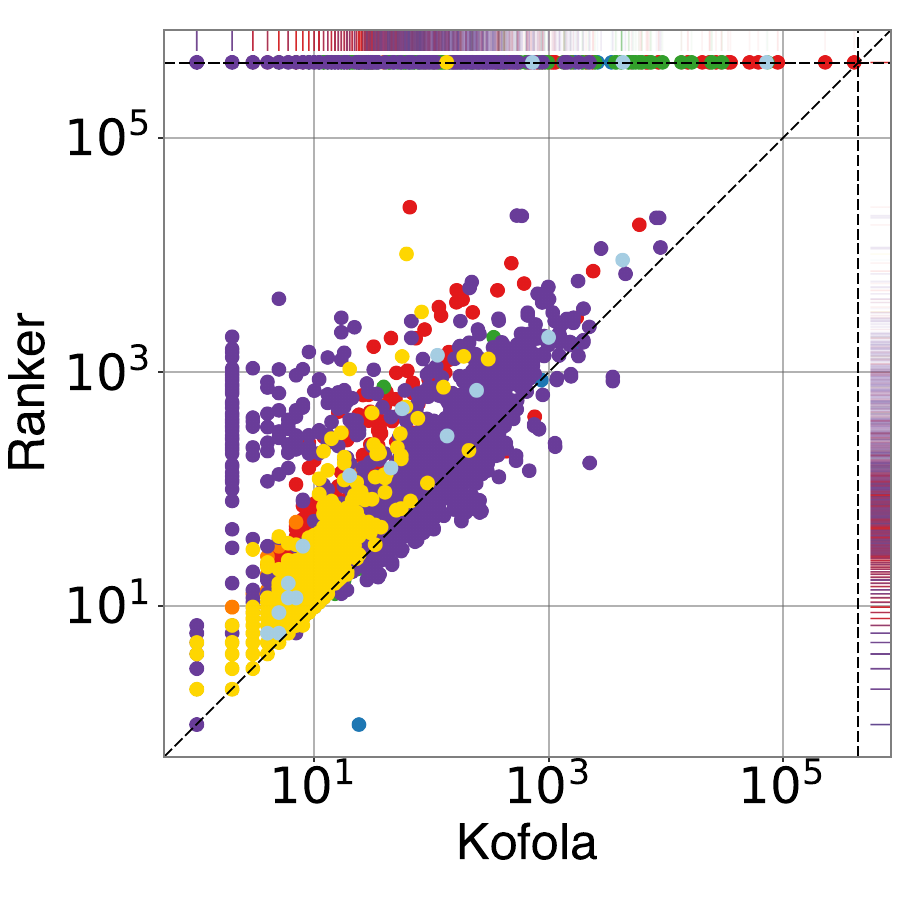}
		\end{subfigure}
		\caption{
      Comparison of state counts of complement automata generated
      by \kofola and other tools; axes are logarithmic.
			Dashed lines are out-of-resources limit and rug plots on margins show distribution of data points.
			Colour indicates benchmark source:
      \stateofbuchi~\RGBcircle{112,60,163}, \seminator~\RGBcircle{243,216,0},
      \automizerbenchc~\RGBcircle{218,42,40}, \ldbaforltl~\RGBcircle{179,211,231},
      \autohyperbenchc~\RGBcircle{54,127,191}, \pecanbenchc~\RGBcircle{76,172,23},
      S1S~\RGBcircle{248,142,0}.
		}\label{fig:scatter-compl}
	\end{figure}
}

\newcommand{\figScatters}[0]{
	\begin{figure}[t]
		\centering
		\begin{subfigure}{0.32\textwidth}
			\centering
			\includegraphics[width=3.8cm,alt={A log-log scatter plot comparing the runtime (in seconds) of inclusion checking by Kofola (x-axis) and Spot(Det) (y-axis), ranging from approximately 0.01 to 100 seconds. Dashed lines mark the out-of-resources limit on both axes. A diagonal dashed line represents equal runtime; points above it indicate Spot(Det) is slower than Kofola, and points below indicate Spot(Det) is faster. Data points are coloured by benchmark source: red (AutoHyperI), blue (Concur), orange (PecanI), and green (AutomizerI).
The majority of points, predominantly green (AutomizerI), cluster in the lower-left region (0.01–1 seconds) near or slightly above the diagonal, suggesting broadly comparable runtimes on those benchmarks with a mild tendency for Spot(Det) to be slower. Several green and red points appear well above the diagonal at low Kofola runtimes (0.01–0.1 s), indicating cases where Spot(Det) is substantially slower. A dense top marginal strip shows many instances — across all benchmark sources — where Spot(Det) hits the out-of-resources limit while Kofola completes within bounds. Conversely, a small number of blue (Concur) and orange (PecanI) points fall below the diagonal at mid-to-high Kofola runtimes (1–100 s), indicating a few cases where Spot(Det) is faster. Overall, Kofola demonstrates a clear runtime advantage over Spot(Det), particularly on AutomizerI and AutoHyperI benchmarks.}]{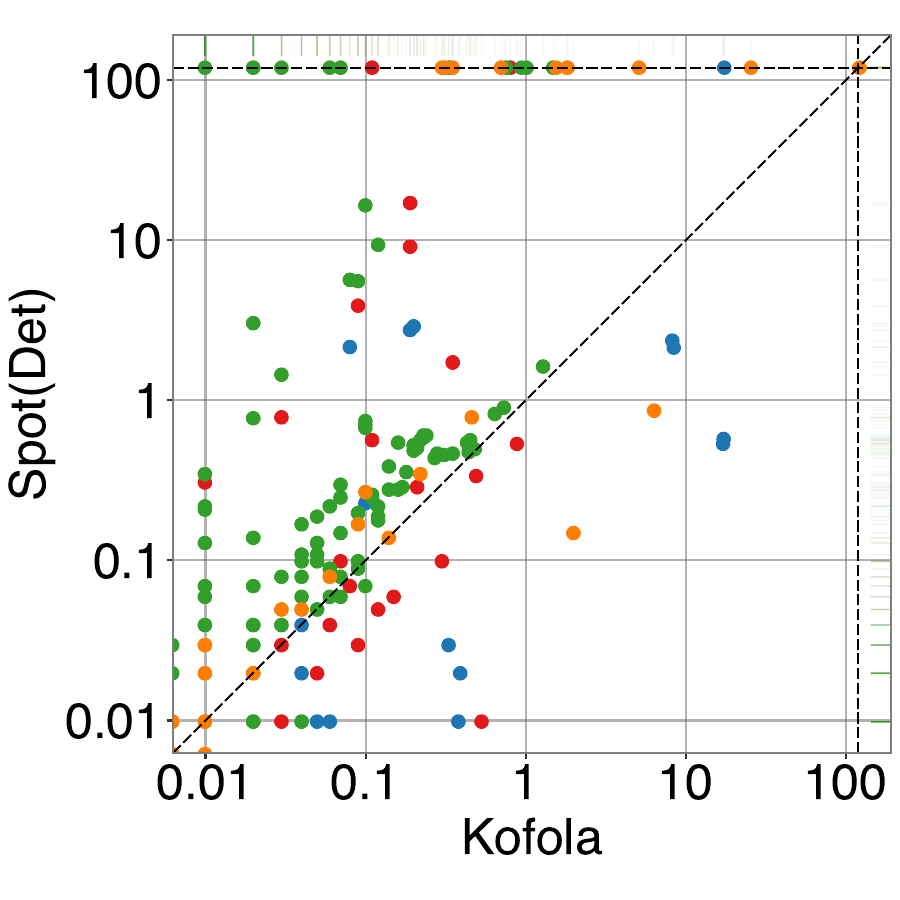}
		\end{subfigure}
		\begin{subfigure}{0.32\textwidth}
			\centering
			\includegraphics[width=3.8cm,alt={A log-log scatter plot comparing the runtime (in seconds) of inclusion checking by Kofola (x-axis) and Forklift (y-axis), ranging from approximately 0.01 to 100 seconds. Dashed lines mark the out-of-resources limit on both axes. A diagonal dashed line represents equal runtime; points above it indicate Forklift is slower than Kofola, and points below indicate Forklift is faster. Data points are coloured by benchmark source: red (AutoHyperI), blue (Concur), orange (PecanI), and green (AutomizerI).
This plot shows the most mixed performance of the runtime comparisons. A dense vertical band of points at very low Kofola runtimes (around 0.01 s) spans a wide range of Forklift values (0.01–100 s), indicating that for many benchmarks Kofola solves instances almost instantly while Forklift takes considerably longer. Red (AutoHyperI) points are notably clustered above the diagonal at low-to-mid Kofola runtimes, suggesting Forklift is consistently slower on those benchmarks. However, a substantial number of green (AutomizerI) and orange (PecanI) points fall below the diagonal at low Kofola runtimes (0.01–0.1 s), indicating cases where Forklift is faster. The right marginal strip shows many green instances where Kofola hits the out-of-resources limit while Forklift completes, in contrast to the top marginal strip where the reverse occurs across multiple benchmark sources. Overall, neither tool dominates clearly — Kofola has an advantage on AutoHyperI benchmarks, while Forklift performs better on a subset of AutomizerI instances.}]{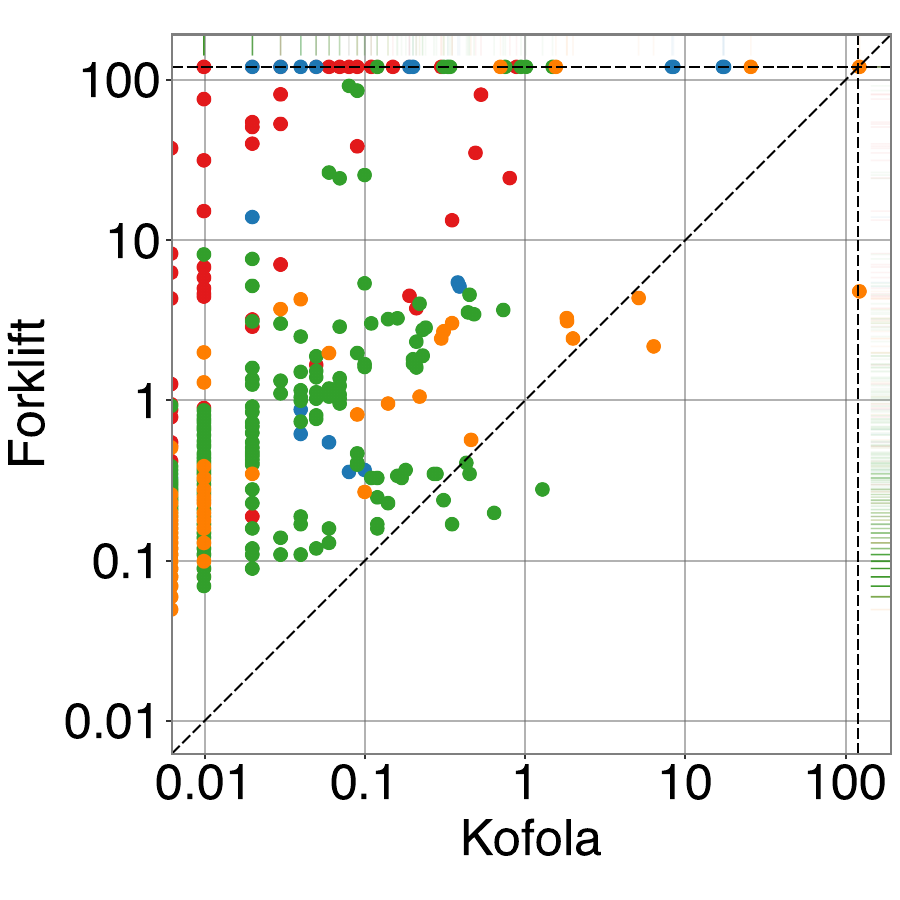}
		\end{subfigure}
		\begin{subfigure}{0.32\textwidth}
			\centering
			\includegraphics[width=3.8cm,alt={A log-log scatter plot comparing the runtime (in seconds) of inclusion checking by Kofola (x-axis) and Rabit (y-axis), ranging from approximately 0.01 to 100 seconds. Dashed lines mark the out-of-resources limit on both axes. A diagonal dashed line represents equal runtime; points above it indicate Rabit is slower than Kofola, and points below indicate Rabit is faster. Data points are coloured by benchmark source: red (AutoHyperI), blue (Concur), orange (PecanI), and green (AutomizerI).
The plot shows a strong overall advantage for Kofola. A very dense vertical band at the lowest Kofola runtimes (around 0.01 s) spans Rabit values from 0.1 to the out-of-resources limit, indicating many benchmarks where Kofola finishes near-instantly while Rabit is orders of magnitude slower or fails entirely. Red (AutoHyperI) points are almost entirely above the diagonal, often reaching the out-of-resources ceiling, indicating Rabit consistently struggles on those benchmarks. Green (AutomizerI) and orange (PecanI) points also predominantly sit above the diagonal. Blue (Concur) points are sparse but notable — several appear below the diagonal at mid-to-high Kofola runtimes (1–20 s), suggesting Rabit is faster on a small number of Concur instances. The right marginal strip contains many green tick marks, reflecting cases where Kofola times out while Rabit completes, though this is clearly the minority pattern. Overall, Kofola is substantially faster than Rabit across most benchmark sources.}]{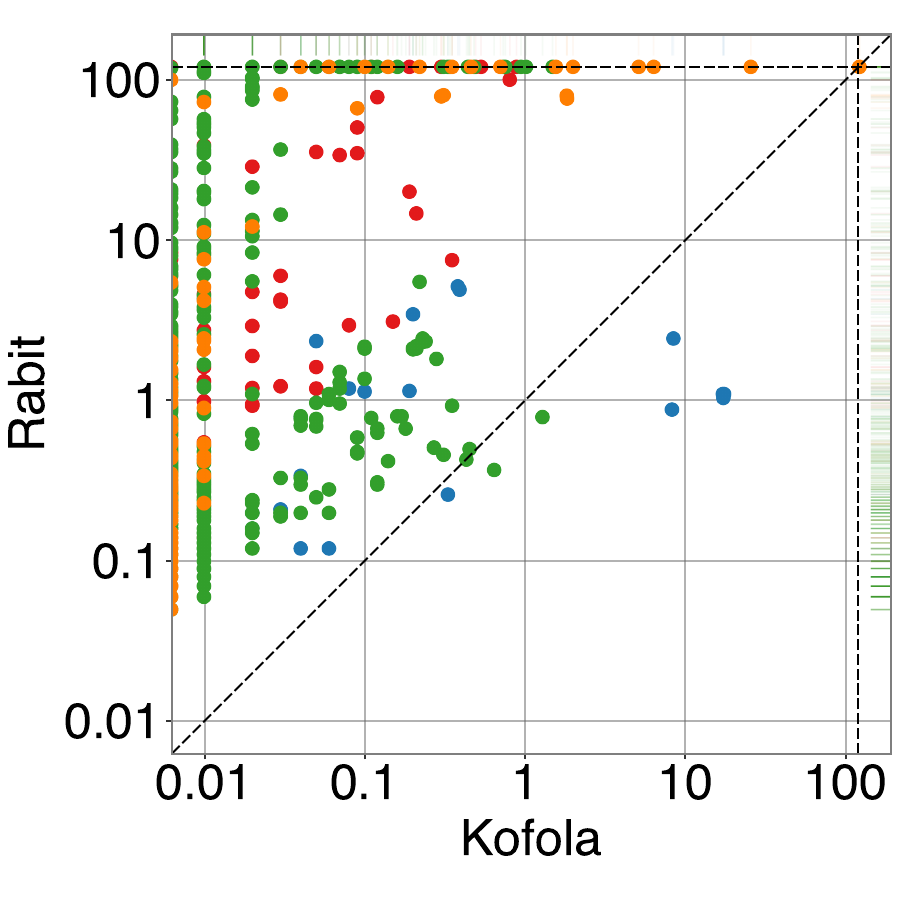}
		\end{subfigure}
		\caption{
			Comparison of runtime of inclusion checking for \kofola and other tools.
			Times are in seconds, axes are logarithmic.
			Dashed lines represent out-of-resources limit and rug plots on margins show distribution of data points.
			The colour of a~point indicates the source of the benchmark:
      \autohyperbenchi~\RGBcircle{223,38,37}, \concur~\RGBcircle{55,129,192},
      \pecanbenchi~\RGBcircle{248,142,0}, \automizerbenchi~\RGBcircle{76,172,23}.
		}\label{fig:scatter-incl}
	\end{figure}
}

\newcommand{\tabStats}[0]{
\begin{table}[t]
\begin{minipage}[t]{0.42\textwidth}
\caption{Statistics for complementation. The column \redfattimes{} has
  numbers of automata (out of 31,273) for which the tool ran out of resources
  (time or memory), \textbf{avg} contains the average number of states of the
  output, \textbf{time} the total time in seconds, and \textbf{unsup} the
  number of input automata with unsupported acceptance condition.}
\label{tab:compl}
\resizebox{\textwidth+2pt}{!}{\hspace*{-10pt}\newcolumntype{h}{>{\columncolor{Gray!30}}r}
\newcolumntype{f}{>{\columncolor{Gray!30}}l}
\newcolumntype{g}{>{\columncolor{Gray!30}}c}
\renewcommand{\arraystretch}{1.15}
\begin{tabular}{lhrhrh}
\toprule
  \multicolumn{1}{c}{\bf tool}             & \multicolumn{1}{g}{~~\redfattimes~~} & \multicolumn{1}{c}{\bf avg}&     \multicolumn{1}{g}{\bf ~~time\,~}& \multicolumn{1}{c}{\bf ~unsup~~} \\
\midrule
\rowcolor{GreenYellow}\kofola & 0   & 79.24    & 1,069& 0 \\
\spot                         & 2   & 115.38   & 199& 0 \\
\kofolaslice                  & 4   & ~~166.03 & 1,376& 0 \\
\kofolaold                    & 44  & 42.74    & 2,471& 0 \\
\owl                          & 104 & 37.39    & 1,240& 507 \\
\ranker                       & 443 & 45.80    & 22,817& 1,336 \\
\bottomrule
\end{tabular}
%
%
}
\end{minipage}
\hfill
\begin{minipage}[t]{0.56\textwidth}
\caption{Statistics for inclusion checking.
  Column \textbf{unsolved} has
  numbers of inclusion problems (out of 1,014) for which the tool ran out of
  resources (time/memory), \textbf{time} contains the total time (in seconds) for the
  solved problems, \textbf{wins}/\textbf{losses} contains the number of
  problems where \kofola strictly won/lost over the tool, and \textbf{missing} contains
  for how many problems we could not convert the inputs into
  \texttt{.ba} files.}
\label{tab:incl}
\resizebox{\textwidth+2pt}{!}{\hspace*{-10pt}\newcolumntype{h}{>{\columncolor{Gray!30}}r}
\newcolumntype{f}{>{\columncolor{Gray!30}}l}
\newcolumntype{g}{>{\columncolor{Gray!30}}c}
\renewcommand{\arraystretch}{1.15}
\begin{tabular}{lhrhrhr}
\toprule
  \multicolumn{1}{c}{\bf tool}  & \multicolumn{1}{g}{\bf ~unsolved~~} & \multicolumn{1}{c}{\bf time}  & \multicolumn{1}{g}{\bf wins} & \multicolumn{1}{c}{\bf ~losses~~} & \multicolumn{1}{g}{\bf ~missing~~} \\
\midrule
\rowcolor{GreenYellow}\kofola & 2                                &   141  &         &     & 0 \\
\spotdet                      & 28                               &   122  & 263     &  56   & 0 \\
\forklift                     & 35                               & 1,273  & 997     &  9    & 8 \\
\spotforq                     & 54                               & 1,134  & 307     &  89   & 0 \\
\rabit                        & 57                               & 4,300  & 996     &  8    & 8 \\
\bait                         & 64                               & ~~1,981  & ~~1,000    & 5     & 8 \\
\bottomrule
\end{tabular}
}
\end{minipage}
\end{table}
}

\section{Experimental Evaluation}\label{sec:experiments}

We have experimentally evaluated the complementation and inclusion checking
capabilities of \kofola and compared it with state-of-the-art tools.
All experiments were run on an Ubuntu GNU/Linux 24.04 virtual machine with
64\,GiB RAM running on an AMD EPYC 9124 CPU.
The timeout was set to 120\,s.
Correctness of answers was evaluated by cross-comparison between the tools.


\subsection{Complementation}\label{sec:label}

\begin{table}[t]
\begin{minipage}[t]{0.42\textwidth}
\caption{Statistics for complementation. The column \redfattimes{} has
  numbers of automata (out of 31,273) for which the tool ran out of resources
  (time or memory), \textbf{avg} contains the average number of states of the
  output, \textbf{time} the total time in seconds, and \textbf{unsup} the
  number of input automata with unsupported acceptance condition.}
\label{tab:compl}
\resizebox{\textwidth+2pt}{!}{\hspace*{-10pt}\newcolumntype{h}{>{\columncolor{Gray!30}}r}
\newcolumntype{f}{>{\columncolor{Gray!30}}l}
\newcolumntype{g}{>{\columncolor{Gray!30}}c}
\renewcommand{\arraystretch}{1.15}
\begin{tabular}{lhrhrh}
\toprule
  \multicolumn{1}{c}{\bf tool}             & \multicolumn{1}{g}{~~\redfattimes~~} & \multicolumn{1}{c}{\bf avg}&     \multicolumn{1}{g}{\bf ~~time\,~}& \multicolumn{1}{c}{\bf ~unsup~~} \\
\midrule
\rowcolor{GreenYellow}\kofola & 0   & 79.24    & 1,069& 0 \\
\spot                         & 2   & 115.38   & 199& 0 \\
\kofolaslice                  & 4   & ~~166.03 & 1,376& 0 \\
\kofolaold                    & 44  & 42.74    & 2,471& 0 \\
\owl                          & 104 & 37.39    & 1,240& 507 \\
\ranker                       & 443 & 45.80    & 22,817& 1,336 \\
\bottomrule
\end{tabular}
%
%
}
\end{minipage}
\hfill
\begin{minipage}[t]{0.56\textwidth}
\caption{Statistics for inclusion checking.
  Column \textbf{unsolved} has
  numbers of inclusion problems (out of 1,014) for which the tool ran out of
  resources (time/memory), \textbf{time} contains the total time (in seconds) for the
  solved problems, \textbf{wins}/\textbf{losses} contains the number of
  problems where \kofola strictly won/lost over the tool, and \textbf{missing} contains
  for how many problems we could not convert the inputs into
  \texttt{.ba} files.}
\label{tab:incl}
\resizebox{\textwidth+2pt}{!}{\hspace*{-10pt}\newcolumntype{h}{>{\columncolor{Gray!30}}r}
\newcolumntype{f}{>{\columncolor{Gray!30}}l}
\newcolumntype{g}{>{\columncolor{Gray!30}}c}
\renewcommand{\arraystretch}{1.15}
\begin{tabular}{lhrhrhr}
\toprule
  \multicolumn{1}{c}{\bf tool}  & \multicolumn{1}{g}{\bf ~unsolved~~} & \multicolumn{1}{c}{\bf time}  & \multicolumn{1}{g}{\bf wins} & \multicolumn{1}{c}{\bf ~losses~~} & \multicolumn{1}{g}{\bf ~missing~~} \\
\midrule
\rowcolor{GreenYellow}\kofola & 2                                &   141  &         &     & 0 \\
\spotdet                      & 28                               &   122  & 263     &  56   & 0 \\
\forklift                     & 35                               & 1,273  & 997     &  9    & 8 \\
\spotforq                     & 54                               & 1,134  & 307     &  89   & 0 \\
\rabit                        & 57                               & 4,300  & 996     &  8    & 8 \\
\bait                         & 64                               & ~~1,981  & ~~1,000    & 5     & 8 \\
\bottomrule
\end{tabular}
}
\end{minipage}
\end{table}

\noindent\textbf{Used Tools.}
For complementation, the default setting of \kofola uses the
determinization-based approach (see \cref{sec:partial-compl}) for
complementing NACs.
We use \kofolaslice to denote the version that complements NACs using the
slice-based construction.
Before outputting a~result, we reduce it using \spot's postprocessor with level
set to \texttt{Low} (this performs
removal of useless states and one run of the simulation-based
reduction~\cite{BustanG00}).
We compared \kofola against state-of-the-art tools for \buchi complementation:
\spot (version~2.14.2)~\cite{Duret-LutzRCRAS22} (using the command line tool \texttt{autfilt --complement}),
\ranker~\cite{HavlenaLS22b},
\owl~\cite{KretinskyMS18} (we ran its determinization into a~parity automaton,
where complementation is trivial), and
the original version of \kofola from~\cite{HavlenaLLST23}, named~\kofolaold hereafter.
\kofolaold uses the same post-processing as \kofola and \ranker uses a~custom
post-processing, which is not exactly the same, but similar to the one in
\spot.

\smallskip
\noindent\textbf{Benchmarks.}
We collected as many BAs as we could, obtaining, in total, 31,273
automata~\cite{automata-benchmarks}.
These come from various sources:
\begin{inparaenum}[(i)]
\item  915 BAs from checking program termination by \textsc{Ultimate Automizer}~\cite{ChenHLLTTZ18} (\automizerbenchc),
\item  6,301 automata from solving the first-order logic of Sturmian words
  by \pecan~\cite{pecan,DBLP:journals/lmcs/HieronymiMOSSS24} (\pecanbenchc; these are not only BAs, but also,
  e.g,. co-\buchi and parity automata),
\item  72 BAs from model checking of HyperLTL properties using \autohyper~\cite{BeutnerF23} (\autohyperbenchc),
\item  370 BAs from deciding S1S formulae~\cite{HavlenaLS21} (\sones),
\item  18 BAs from LTL to limit-deterministic BA translation~\cite{SickertEJK16}~(\ldbaforltl),
\item  1,721 BAs from translation from real-world and random LTL formulae~\cite{BlahoudekDS20}~(\seminator), and
\item  21,876 randomly generated BAs~\cite{TsaiFVT14}~(\stateofbuchi).
\end{inparaenum}

\smallskip
\noindent\textbf{Results.}
We compared the sizes of the complement automata that the tools generated.
We let the tools output automata with any acceptance condition (\kofola and
\spot could exploit this to obtain smaller automata, whereas \ranker always
generates automata with the \buchi condition).
Statistics of the experiment are in \cref{tab:compl} and interesting scatter plots
are in \cref{fig:scatter-compl} (the scatter plot for \owl looks similar as the
one for \spot with more unsolved cases).
We do not include the median, which was~2~for all tools, in the table.
We can see that \kofola was the only tool that could complement all input
automata.

The second tool, \spot, could complement all but two automata (it ran out of
memory for those), however the average size of the output is much bigger.
We can see this also in the plot comparing the two tools, where \kofola can get
much smaller complements for many \stateofbuchi automata.
For complementation, \kofolaslice gives worse average results than \kofola.
There are, however, cases where \kofolaslice outputs a~smaller
result.
The improvement of \kofola over \kofolaold can be seen mainly in the number of
solved cases.
While the average size is better for \kofolaold, this should be considered in
the context that \kofola can solve the 44 benchmarks that \kofolaold
could not, often with a~large output, which skews the average.
From the plot, we can see that \kofolaold could in many cases output a~slightly
smaller automaton than \kofola.  This is due to the settings of the algorithms,
which we optimized in \kofola to maximize the number of solved cases; as
a~consequence, sometimes the algorithms perform slightly worse.
The dominance of \kofola over \ranker and \owl is clear from the number of solved cases
and the scatter plot (\ranker has a~lower average since
it cannot solve harder benchmarks).
1,336 (for \ranker) and 507 (for \owl) input automata (all from \pecan) used an
unsupported acceptance condition.

Our procedure for complementing IADACs helped in 225 cases, reducing
the number of states of the output by an average of 5 (among the cases where it
had effect).
The most significant improvement was observed on one automaton from the
\seminator benchmark, where the number of states decreased by 236.


\figScatterCompl  

\subsection{Inclusion Checking}\label{sec:experiments:inclusion}

\noindent\textbf{Used Tools.}
For inclusion, \kofola uses the slice-based
approach (\cref{sec:partial-compl}) for complementing NACs, which is more
suitable for on-the-fly exploration.
We evaluated \kofola against state-of-the-art inclusion checkers
\bait~\cite{DoveriGPR21}, \forklift~\cite{forklift}, \rabit~\cite{DBLP:journals/lmcs/ClementeM19}
v2.5.1 with argument \texttt{-fast}, and \spot 2.14.2~\cite{Duret-LutzRCRAS22}
(using \texttt{autfilt --included-in}).
For \spot, we tried two versions:
determinization-based~\cite{Safra88,Piterman07,Redziejowski12} (\spotdet)
and well-quasiorder-based~\cite{forklift} (\spotforq).

\smallskip
\noindent\textbf{Benchmarks.}
We used inclusion problems involving BAs given in the \texttt{HOA}
format~\cite{BabiakBDKKM0S15} from all available sources we are aware of
that are practically motivated (the
number of automata pairs is given in parentheses):
\begin{inparaenum}[(i)]
  \item inclusion problems emerging from model checking of hyperproperties
    using the tool \autohyper~\cite{BeutnerF23}, denoted as \autohyperbenchi (36; max.\ 66,051,
    avg.\ 2,059 states),
  \item instances occurring in program termination checking using
  \textsc{Ultimate} \automizer~\cite{ChenHLLTTZ18}, denoted as \automizerbenchi (404; max.\ 88,304, avg.\ 904 states), 
  \item inclusion problems coming from verification of concurrent
    systems~\cite{AbdullaCCHHMV11} (13~instances denoted as \concur, omitting one benchmark from the source due
    to a~missing counterpart automaton; max.\ 1,532, avg.\ 589 states), and
  \item logical implication tasks in word combinatorics from
    \pecan~\cite{pecan,DBLP:journals/lmcs/HieronymiMOSSS24}, denoted as \pecanbenchi (54 instances where it was clear how
    to pair them; max.\ 712,184, avg.\ 7,729 states).
\end{inparaenum}
%
Without performing any pre-selection on the existing benchmarks of the literature, 98\,\% of the 507 BAs appearing on the right-hand side of the inclusion (i.e., those that are being complemented) are elevator automata.
This emphasizes the importance of exploiting automata structure in inclusion checking.
From the automata on the left-hand side, 76\,\% were elevator automata (in
total, 87\,\% of the involved BAs were elevator).
To get more interesting examples for inclusion checking, we also
swapped the sides of the benchmarks for a~total of 1,014 inclusion problems
(the vast majority of these benchmarks were non-trivial, often harder than the
original problems).
For tools not supporting the \texttt{HOA} format (\bait,
\forklift, and \rabit), we employed \spot to convert the input automata into
the \texttt{.ba} format.
For four of the original inclusion problems (eight benchmarks in total), this was not possible (the translation took too many resources).
Overall, inclusion holds in 12\,\% of the benchmarks.


\figScatters  


\smallskip

\noindent\textbf{Results.}
We give statistics for inclusion checking in \cref{tab:incl} and interesting scatter plots in
\cref{fig:scatter-incl}.
\kofola won by far on the number of solved instances and the total time.
The second tool was \spotdet, whose time is slightly lower, but solved 26 less instances.
The scatter plot shows that it can sometimes beat \kofola.
Ramsey-based tools follow with
run times being an order of magnitude worse.
The best of them is \forklift (surprisingly to us, it solves more problems
than \spotforq) with \spotforq, \rabit, and \bait trailing.
From the scatter plots (we omit \bait, which mostly
behaves worse than \forklift and \spotforq), the tool most complementary to \kofola is
\spot.


\begin{credits}
\subsubsection{\ackname}
We thank the anonymous reviewers for their constructive feedback.
This work was supported in part by
the National Natural Science Foundation of China (Grant No.\ 62102407), the CAS Project for Young Scientists in Basic Research (Grant No.\ YSBR-040),
the Czech Science Foundation projects \mbox{26-22640S} and \mbox{25-17934S}, and
the FIT BUT internal project FIT-S-26-9011.
\subsubsection{\discintname}
The authors have no competing interests.
\subsubsection{Data Availability Statement.}
An environment with the tools and data used for the experimental evaluation in
the current study is available at~\cite{artifact}.
\end{credits}

\bibliographystyle{splncs04}
\bibliography{literature.bib}

\clearpage
\appendix 

\crefalias{section}{appendix}

\end{document}